\documentclass[authoryear]{article}

\usepackage[left=2.25cm,right=2.25cm,top=1.0cm,bottom=2.0cm,papersize={182mm,257mm}]{geometry}

\usepackage{natbib}
\usepackage{graphicx}
\usepackage{color}
\usepackage{soul}
\usepackage{tikz}
\usepackage{url}
\usepackage{algorithm}
\usepackage{algpseudocode}
\usepackage{physics}
\usepackage{amssymb}
\usepackage{amsmath}
\usepackage{ulem}
\usetikzlibrary{patterns}
\usepackage{cases}
\usepackage{authblk}

\newtheorem{property}{Property}
\newtheorem{prop}{Proposition}
\newtheorem{thm}{Theorem}

\newtheorem{proof}{Proof}

\begin{document}
\title{Single machine rescheduling for new orders: properties and complexity results}
\date{}
\author[1]{Elena Rener} 
\author[1]{Fabio Salassa}
\author[1,2]{Vincent T'kindt}

\affil[1]{DIGEP, Politecnico di Torino, Torino, Italy}
\affil[2]{Laboratoire d'Informatique Fondamentale et Appliqu\'ee, Universit\'e de Tours, Tours, France}

\providecommand{\keywords}[1]{\textit{Keywords:} #1}

\maketitle 
\begin{abstract}
	Rescheduling problems arise in a variety of situations where a previously planned schedule needs to be adjusted to deal with unforeseen events. A common problem is the arrival of new orders, i.e. jobs, which have to be integrated into the schedule of the so-called old jobs. The maximum and total absolute time deviations of the completion times of these jobs are modeled as a disruption constraint to limit the change in the original schedule. Disruption constraints affect the shape of an optimal schedule, particularly with respect to the sequencing of old jobs and the insertion of idle time. We therefore give a classification into idle and no-idle problems for a set of single-machine rescheduling problems with different objective functions. We then prove the complexity of five rescheduling problems that have been left open in the literature. 
\end{abstract}

\section{Introduction}
In the literature, rescheduling has been considered as a way of modifying an existing schedule to respond to unforeseen disruptive events, such as the arrival of new jobs, the delay of some other jobs, machine breakdowns, or periods of unavailability of machines due to maintenance activities.
The aim of the present work is to study scheduling strategies when unexpected jobs, called \textit{new jobs}, have to be scheduled in an already given optimal solution made up of known jobs, called \textit{old jobs}. This is typically called rescheduling for new orders.
These strategies seek a trade-off between limiting the perturbation of the original schedule, called \textit{disruption} of the schedule, and integrating the new jobs to optimize an objective function. \\

Rescheduling problems for new orders have been the subject of a number of studies. 
This family of rescheduling problems was first formalized by \cite{HP:04} and has since received much attention. \cite{HP:04} studied several single machine scheduling problems. 
They considered total completion time or maximum lateness as objective functions to be minimized in combination with a disruption cost. They provided polynomial-time algorithms or show $\cal NP$-hardness for each of the problems tackled. The same problems were studied by \cite{TT14}, who considered the enumeration of Pareto optima. \cite{ZLY16} introduced a generalization of the disruption criterion by imposing on each old job $j$ a maximum amount of allowed disruption $k_j$. In this case, both the problems of minimizing the sum of the completion times and the maximum lateness become $\cal NP$-hard.
Only two research papers were found that consider rescheduling problems with multiple disruptions (\cite{HLP:07} and \cite{YMLL:07}). The setting here changes in the sense that there are multiple sets of new jobs arriving at different times. Thus, the initial state does not necessarily imply optimal schedules. In addition to multiple disruptions, \cite{YMLL:07} also considered jobs with release dates.
\cite{YM:07} studied four single machine rescheduling problems with minimization of the makespan in the presence of job release dates and different disruption criteria. Complexity results were given for the different scenarios.
\cite{BY07}, \cite{ZT10} and \cite{LZ15} considered the single machine rescheduling problem when jobs have time-dependent processing times.
\cite{guo_single-machine_2021} studied a rescheduling problem motivated by energy savings. New orders are rework jobs that need to be added to the original sequence to minimize the total waiting time of the jobs. \cite{liu_cost_2018} investigated the rescheduling on a two-machine flow shop with outsourcing solved by means of a hybrid variable neighbourhood search.
\cite{zhang_rescheduling_2021} took into consideration the minimization of the maximum weighted tardiness for rescheduling problems. \cite{rener_single_2022} studied the rescheduling problem for minimizing the maximum lateness with a constraint on the total time disruption and proposed the first exact algorithm to solve the problem when a no-idle constraint is also required. 
Finally, \cite{fang_rescheduling_2023} studied the rescheduling problem with rejection for minimizing a linear combination of the total weighted completion time, the maximum disruption and the rejection cost. They first proved the $\cal NP$-hardness of the problem, and then solved it with an exact dynamic programming algorithm and a fully polynomial time approximation algorithm. They showed the efficiency of the proposed algorithms with several computational results. 

We refer the readers to the surveys by \cite{VHL:03} and \cite{aytug_executing_2005} for a general introduction to rescheduling problems. More recent papers covering the topic from different perspectives are proposed for instance by \cite{pferschy_algorithms_2022}, \cite{luo_tardiness-augmented_2022}, \cite{wang_2018}.\\

In this paper, we consider rescheduling problems for new orders arising from the combination of two disruption criteria and several scheduling objective functions. Our contributions, which focus on single machine problems, can be summarized as follows:
\begin{itemize}
	\item we provide general properties for rescheduling with new orders to minimize classical scheduling objective functions, subject to a constraint on the maximum or total time disruption of the original jobs;
	\item we propose a classification of rescheduling problems according to whether machine idle times are to be inserted in optimal solutions;
	\item we provide proofs of the computational complexity of five open problems.
\end{itemize}

The remainder of the paper is organized as follows: in section 2, we formally define rescheduling problems and give contributions regarding the structure of optimal solutions. We present both properties for preserving the order of old jobs after rescheduling and properties on the requirement of inserting idle times. In section 3, we define the computational complexity of some open problems. 
Section 4 gives the conclusions and directions for future work.

\section{Structural properties}
\subsection{Notation}
Let $J^O$ (resp. $J^N$) be the set of $n_O$ old jobs (resp. $n_N$ new jobs) that must be processed on a single machine. Each job $j$ has a processing time $p_j$ and, depending on the problem, may have a due date $d_j$ or a weight $w_j$. Given any schedule $\sigma$, let $C_j(\sigma)$ ($C_j$ when there is no ambiguity) be the completion time of job $j$. Initially, jobs in $J^O$ are assumed to be optimally scheduled in a sequence $\pi^*$ to minimize a given objective function $f$. After the arrival of  jobs from $J^N$, the goal is to generate a new schedule $\sigma^*$ of all jobs to minimize $f$ while not disrupting too much $\pi^*$. The disruption of the schedule is modeled as the maximum or the sum of the absolute deviations of the completion times of old jobs. Let $\Delta_{max}=max_{j\in J^O}(\Delta_j)$ be the maximum disruption and $\overline \Delta = \sum_{j\in J^O}\Delta_j$ be the total disruption, where the disruption is given by $\Delta_j=|C_j(\sigma^*)-C_j(\pi^*)|, \forall j\in J^O$. A threshold $\epsilon$ limits the amount of disruption.  
Each job in a schedule is associated with a function $f_j$, which we assume to be regular, i.e. increasing with respect to the job completion time. 
We consider, for each problem, an objective function $f\in \{f_{max},\overline f\}$, where $f_{max}=\max_{j\in J^O \cup J^N}f_j$ and $\overline f=\sum_{j\in J^O \cup J^N}f_j$. 
Whenever interesting, we will consider the following particular cases of $\overline f$ and $f_{max}$. Among the $\overline f$ objectives, we consider three different objective functions. When considering the total weighted completion time $\overline{wC}$, for each job $f_j=w_jC_j$. For computing the total number of late jobs $\overline{U}$, $f_j=1$ if $C_j>d_j$ and $f_j=0$ otherwise. Finally, when the objective is the total tardiness $\overline{T}$, $f_j=\max(0,C_j-d_j)$.
As a $f_{max}$ objective, we consider the most widely used in the scheduling literature, which is the maximum lateness $L_{max}$, where $f_j=C_j-d_j$. 

Using the classical three-field notation of \cite{GLLR:79}, the problems we consider in this work are rescheduling problems on a single machine in the form of $1|\Delta_{max}\leq\epsilon|f$ and 
$1|\overline \Delta\leq \epsilon|f$. 

\subsection{Preserving the initial order for old jobs}
Old jobs are initially sequenced in a schedule $\pi^*$ to minimize an objective function $f$. This order might be preserved in an optimal schedule of the rescheduling problem.
In this section, we focus on establishing some properties of when this order is preserved, as opposed to when a \textit{re-sequencing} of old jobs is required.\\
The first two properties illustrated below serve a dual purpose. The first is to provide the source of the need for idle time insertion. The second is to provide actual structural properties for the complementary problems $1|f\leq \epsilon|\Delta_{max}$ and $1|f\leq \epsilon|\overline \Delta$. \\
The two disruption criteria lead to different behaviours. In fact, we first show that problem $1|f\leq \epsilon|\Delta_{max}$ never requires a re-sequencing to obtain an optimal solution, if any feasible schedule exist with old jobs ordered as in $\pi^*$. On the contrary, we show that problem $1|f\leq \epsilon|\overline \Delta$ may require a re-sequencing of old jobs, even if there exists a feasible schedule with old jobs ordered as in $\pi^*$.

\begin{property}
	\label{prop::deltamax}
	If there exists a non-empty set of feasible schedules for problem $1|f\leq \epsilon|\Delta_{max}$ with old jobs ordered as in $\pi^*$, there exists an optimal solution that belongs to this set of schedules.
\end{property} 

\begin{figure}[h]
	\centering
	\begin{tikzpicture}[x=0.4cm,y=0.4cm]
		\draw[xstep=1,ystep=1,ultra thin](0,0)(13,1);
		\draw[](-1.4,0.5)node{ $\pi^*$};
		\draw[](1.7,0.5)node{\small $i$};
		\draw[](4.5,0.5)node{\small $j$};
		\draw[](0.6,0.5)node{\small $\dots$};
		\draw[](3,0.5)node{\small $\dots$};
		\draw[](6,0.5)node{\small $\dots$};
		\draw[->](0,0)--(11,0)node[right]{\small $t$};
		\draw[](1.2,0)rectangle (2.2,1);
		\draw[](0,0)node[left]{0}rectangle(7,1);
		\draw[](4,0)rectangle(5,1);
	\end{tikzpicture}\\	[2ex]
	\begin{tikzpicture}[x=0.4cm,y=0.4cm]
		\draw[xstep=1,ystep=1,ultra thin](0,0)(13,1);
		\draw[](-1.4,0.5)node{ $\sigma'$};
		\draw[](3.5,0.5)node{\small $j$};
		\draw[](8.5,0.5)node{\small $i$};
		\draw[](1.5,0.5)node{\small $\dots$};
		\draw[](6,0.5)node{\small $\dots$};
		\draw[](9.7,0.5)node{\small $\dots$};
		\draw[->](0,0)--(11,0)node[right]{\small $t$};
		\draw[](0,0)node[left]{0}rectangle(10.5,1);
		\draw[](3,0)rectangle(4,1);
		\draw[](8,0)rectangle(9,1);
	\end{tikzpicture}\\ [2ex]\begin{tikzpicture}[x=0.4cm,y=0.4cm]
		\draw[xstep=1,ystep=1,ultra thin](0,0)(13,1);
		\draw[](-1.4,0.5)node{ $\sigma^*$};
		\draw[](3.5,0.5)node{\small $i$};
		\draw[](8.5,0.5)node{\small $j$};
		\draw[](1.5,0.5)node{\small $\dots$};
		\draw[](6,0.5)node{\small $\dots$};
		\draw[](9.7,0.5)node{\small $\dots$};
		\draw[->](0,0)--(11,0)node[right]{\small $t$};
		\draw[](0,0)node[left]{0}rectangle(10.5,1);
		\draw[](3,0)rectangle(4,1);
		\draw[](8,0)rectangle(9,1);
	\end{tikzpicture}
	\caption{Jobs sequencing for the $1|f\leq \epsilon|\Delta_{max}$ problem.}
	\label{fig:deltamax_ordering}
\end{figure}
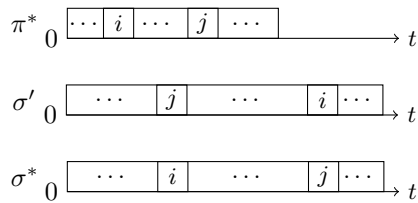
\begin{proof}
	Consider Figure \ref{fig:deltamax_ordering}. Schedule $\pi^*$ shows the original optimal schedule with old jobs only. Schedule $\sigma'$ is a schedule with old and new jobs, where at least two old jobs $i$ and $j$ are swapped with respect to their ordering in $\pi^*$. Suppose that $\sigma'$ is an optimal solution for the problem $1|f\leq \epsilon|\Delta_{max}$. Let $i$ be the first old job that has been moved after other old jobs that followed $i$ in $\pi^*$, and let $j$ be the old job that precedes $i$ in $\sigma'$ w.r.t old jobs only. Notice that in $\sigma'$, before job $j$ there can be both old and new jobs, while between $j$ and $i$ there can be only new jobs because $j$ is the closest old job that precedes $i$ in $\sigma'$.
	Let  $\sigma^*$ be the same schedule than $\sigma'$ but with $i$ and $j$ swapped, such that the starting time of $j$ in $\sigma'$ is the same of $i$ in $\sigma^*$ and $C_i(\sigma')=C_j(\sigma^*)$. In $\sigma^*$, before job $i$ there can be both old and new jobs, while between $i$ and $j$ there can be new jobs only. Since $\sigma'$ is assumed to be optimal, $\Delta_{max}(\sigma')\leq \Delta_{max}(\sigma^*)$. 
	
	The proof is done by contradiction showing that, if there exists such a schedule $\sigma^*$, with $f\leq \epsilon$, then  $\Delta_{max}(\sigma^*)\leq \Delta_{max}(\sigma')$ and $\sigma'$ cannot be optimal.
	
	First, notice that $\Delta_i\geq0$ in both schedules by definition of job $i$ and that $\Delta_i(\sigma^*)< \Delta_i(\sigma')$ since $C_i(\sigma')>C_i(\sigma^*)$. 
	
	Given that $C_j(\pi^*)>C_i(\pi^*)$ and $C_i(\sigma')=C_j(\sigma^*)$, we have that $\Delta_j(\sigma^*)< \Delta_i(\sigma')$. 
	
	Next, we show that $\Delta_j(\sigma')\leq \Delta_i(\sigma')$. If $C_j(\sigma')\geq C_j(\pi^*)$, then this holds, since $i$ is more right-shifted than $j$. If $C_j(\sigma')< C_j(\pi^*)$, we have $\Delta_j(\sigma')\leq C_j(\sigma')-C_j(\sigma^*)$. For the disruption of job $i$ it holds $\Delta_i(\sigma')\geq C_i(\sigma')-C_i(\sigma^*)$ and since $C_j(\sigma')-C_j(\sigma^*)=C_i(\sigma')-C_i(\sigma^*)$, we have $\Delta_j(\sigma')\leq C_j(\sigma')-C_j(\sigma^*)=C_i(\sigma')-C_i(\sigma^*)\leq \Delta_i(\sigma')$.
	
	Putting all together, we obtain
	$$ \max(\Delta_i(\sigma'),\Delta_i(\sigma')) = \Delta_i(\sigma') > \max(\Delta_i(\sigma^*),\Delta_i(\sigma^*)),  $$
	and since all other old jobs do not change their starting and completion times,
	$$ \Delta_{max}(\sigma')>\Delta_{max}(\sigma^*).$$
	
	Repeating the argument for any pair of old jobs $i$ and $j$ defined as above proves the statement.
\end{proof}


Now consider the other complementary problem $1|f\leq \epsilon|\overline \Delta$. 
The following property holds.
\begin{property}
	\label{prop:sumdelta}
	For problems $1|f\leq \epsilon|\overline \Delta$, an optimal schedule may have old jobs scheduled in a different order with respect to  $\pi^*$ even if swapping them keeps the schedule feasible with respect to the constraint $f\leq \epsilon$.
\end{property} 

\begin{proof}
	\begin{figure}[h]
		\centering
		\begin{tikzpicture}[x=0.2cm,y=0.4cm]
			\draw[xstep=1,ystep=1,ultra thin](0,0)(22,1);
			\draw[](-3,0.5)node{$\pi^*$};
			\draw[](-0.6,0)node{\small $0$};
			\draw[](4,0.5)node{\small $i$};
			\draw[](9,0.5)node{\small $j$};
			\draw[](10,-0.7)node{\small $10$};
			\draw[->](0,0)--(15,0)node[right]{\small $t$};
			\draw[-](0,0)--(0,1);
			\draw[-](0,1)--(10,1);
			\draw[-](10,1)--(10,0);
			\draw[-](8,0)--(8,1);
		\end{tikzpicture}\\ [2ex]
		
		\begin{tikzpicture}[x=0.2cm,y=0.4cm]
			\draw[xstep=1,ystep=1,ultra thin](0,0)(22,1);
			\draw[](-3,0.5)node{$\sigma'$};
			\draw[](-0.6,0)node{\small $0$};
			\draw[](16,0.5)node{\small $i$};
			\draw[](11,0.5)node{\small $j$};
			\draw[](5,0.5)node{\small $\dots$};
			\draw[](10,-0.7)node{\small $10$};
			\draw[->](0,0)--(22,0)node[right]{\small $t$};
			\draw[](10,0)rectangle(20,1);
			\draw[](10,0)rectangle(12,1);
			\draw[](12,0)rectangle(20,1);
		\end{tikzpicture}\\ [2ex]
		
		\begin{tikzpicture}[x=0.2cm,y=0.4cm]
			\draw[](-3,0.5)node{$\sigma^*$};
			\draw[xstep=1,ystep=1,ultra thin](0,0)(22,1);
			\draw[](-0.6,0)node{\small $0$};
			\draw[](14,0.5)node{\small $i$};
			\draw[](19,0.5)node{\small $j$};
			\draw[](5,0.5)node{\small $\dots$};
			\draw[](10,-0.7)node{\small $10$};
			\draw[->](0,0)--(22,0)node[right]{\small $t$};
			\draw[](10,0)rectangle(20,1);
			\draw[](10,0)rectangle(18,1);
			\draw[](18,0)rectangle(20,1);
		\end{tikzpicture}
		\caption{Jobs sequencing in the $1|f\leq \epsilon|\overline \Delta$ problem.}
		\label{fig:deltasum_ordering}
	\end{figure}
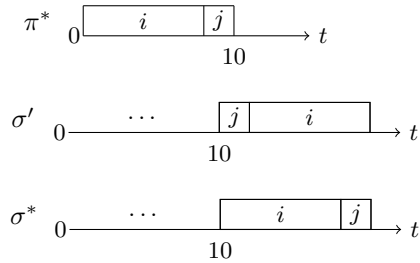
	The proof is done by showing a counterexample.
	Consider two old jobs $i$ and $j$ with $p_i=8$ and $p_j=2$ scheduled consecutively so that $C_i(\pi^*)=8$ and $C_j(\pi^*)=10$ (Figure \ref{fig:deltasum_ordering}). Let be $d_i=C_i(\pi^*)$ and $d_j=C_j(\pi^*)$. 
	
	Schedule $\sigma'$ is a schedule with old and new jobs, where the two old jobs are swapped with respect to their ordering in $\pi^*$. Suppose that $\sigma'$ is optimal for solving $1|f\leq \epsilon|\overline \Delta$. 
	Let be $\sigma^*$ a feasible schedule, with all jobs, except jobs $i$ and $j$, which are swapped, scheduled as in $\sigma'$.
	Assume that the starting time of job $i$ in $\sigma'$ as well as the starting time of job $j$ in $\sigma^*$ is $t=10$.
	Then, the optimal sequencing of the jobs is  $(j,\,i)$ with a resulting $\overline \Delta(\sigma^*)=14 < \overline \Delta(\sigma')=20$. Since any other new job scheduled does not have an effect on the total disruption, schedule $\sigma'$ is suboptimal.
\end{proof} 

\subsection{Relevance of idle-time insertion}
Usually, minimizing a regular objective function allows to consider only the sequencing or permutation problem in scheduling. However, for our rescheduling problems, the constraint on the total disruption implicitly models a second objective, that is non regular. So, inserting idle times becomes relevant for solving the problem. Several considerations and examples are given to provide insights on when the insertion of machine idle times is required.

Consider the following example. Three old jobs are originally scheduled in $\pi^*=(i,j,k)$. Let be $p_i=6,p_j=p_k=1$. Assume that the minimum cost schedule with no inserted idle times, that includes a new job $h$, with $p_h=1$, is $\sigma^*_{no-idle}=(h,j,k,i)$ as in Figure \ref{fig:disrup_decrease}. In schedule $\sigma^*_{no-idle}$, we have $C_j(\pi^*)-C_j(\sigma^*_{no-idle})>C_i(\sigma^*_{no-idle})-C_i(\pi^*)$ and $C_k(\pi^*)-C_k(\sigma^*_{no-idle})>C_i(\sigma^*_{no-idle})-C_i(\pi^*)$. 
In this case, constructing a schedule $\sigma^*_{idle}$ with one unit of idle time before job $j$ reduces both $\Delta_{max}$ and $\overline \Delta$. So, in this example, for the $\Delta_{max}$ criterion and $\epsilon=5$, schedule $\sigma^*_{idle}$ makes this sequence of jobs feasible, and possibly optimal, as well as for the $\overline \Delta$ criterion and $\epsilon=14$.
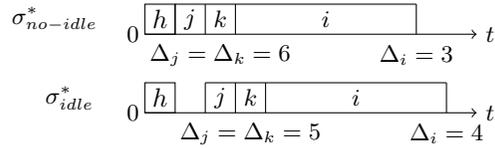
\begin{figure}[h]
	\centering
	\begin{tikzpicture}[x=0.4cm,y=0.4cm]
		\draw[xstep=1,ystep=1,ultra thin](-5,0)(10,1);
		\draw[](-3,0.5)node{\small $\sigma^*_{no-idle}$};
		\draw[](-0.4,0)node{\small $0$};
		\draw[](6,0.5)node{\small $i$};
		\draw[](1.5,0.5)node{\small $j$};
		\draw[](2.5,0.5)node{\small $k$};
		\draw[](9,-0.7)node{\small $\Delta_i=3$};
		\draw[](2.5,-0.7)node{\small $\Delta_j=\Delta_k=6$};
		\draw[->](0,0)--(11,0)node[right]{\small $t$};
		\draw[-](1,0)--(1,1);
		\draw[-](1,1)--(9,1);
		\draw[-](9,1)--(9,0);
		\draw[-](2,0)--(2,1);
		\draw[-](3,0)--(3,1);
		\draw[](0,0)rectangle(1,1);
		\draw[](0.5,0.5)node{$h$};
	\end{tikzpicture}\\
	
	\begin{tikzpicture}[x=0.4cm,y=0.4cm]
		\draw[xstep=1,ystep=1,ultra thin](-5,0)(10,1);
		\draw[](-2.5,0.5)node{\small $\sigma^*_{idle}$};
		\draw[](-0.4,0)node{\small $0$};
		\draw[](7,0.5)node{\small $i$};
		\draw[](2.5,0.5)node{\small $j$};
		\draw[](3.5,0.5)node{\small $k$};
		\draw[](10,-0.7)node{\small $\Delta_i=4$};
		\draw[](3.5,-0.7)node{\small $\Delta_j=\Delta_k=5$};
		\draw[->](0,0)--(11,0)node[right]{\small $t$};
		\draw[-](2,0)--(2,1);
		\draw[-](2,1)--(10,1);
		\draw[-](10,1)--(10,0);
		
		\draw[-](3,0)--(3,1);
		\draw[-](4,0)--(4,1);
		\draw[](0,0)rectangle(1,1);
		\draw[](0.5,0.5)node{$h$};
	\end{tikzpicture}
	\caption{Decrease of the schedule disruption via idle time insertion.}
	\label{fig:disrup_decrease}
\end{figure}

The example gives a first feeling of the correlation between the re-sequencing of old jobs and the insertion of idle times to obtain optimal solutions. Notice that the decrease of the disruption by right-shifting is in fact due by the fact that jobs $j$ and $k$ were scheduled before their original completion time, and $i$ after.

The following property applies to an optimal schedule, where the old jobs are sequenced as in the original schedule. 
\begin{property}
	\label{prop::seq_old}
	If there exists an optimal solution in which old jobs are ordered as in the original schedule $\pi^*$, then this solution is without inserted idle times.
\end{property}
\begin{proof} 
	Since jobs in $J^O$ are in the same order as in schedule $\pi^*$, the disruption of any job $j$ in an optimal schedule $\sigma^*$ is given by $\Delta_j=\sum_{i\in J^N: i\to j}p_i$, where $i\in J^N: i\to j$ denote the new jobs $i$ that precede $j$ in the schedule. By inserting an idle time $\delta> 0$ before job $j$, the disruption changes to $\Delta'_j=\sum_{i\in J^N: i\to j}p_i + \delta$. Then, in this new schedule $\sigma'$, $\Delta_{max}(\sigma')\geq \Delta_{max}(\sigma^*)$ as well as $\overline \Delta(\sigma')> \overline \Delta(\sigma^*)$.\\
	Also, $f(C_j+\delta)>f(C_j),\,\ \forall j\in J^O\cup J^N$ since the $f_j$'s are regular functions. Then, $f_{max}(\sigma')\geq f_{max}(\sigma^*)$ as well as $\overline f(\sigma')> \overline f(\sigma^*)$.\\
	Hence, inserting any idle time does not decrease the considered objective functions nor the disruption criteria.
\end{proof}


Let us now consider some rescheduling problems with particular objective functions. Hereafter, we consider the particular cases $f_{max} = L_{max}$ and $\overline f \in \{\overline C; \overline{wC}; \overline U; \overline T\}$. 

\begin{prop} \label{prop:no-idle_prbs}
	\textup{\textbf{(\cite{HP:04}, \cite{lendl})}} There exists an optimal solution to problems $1|\Delta_{max}\leq \epsilon|f$, $f\in \{L_{max}, \overline C, \overline{wC}\}$, and $1|\overline \Delta\leq \epsilon|\overline C$, where the jobs are processed consecutively without inserted idle times.
\end{prop}
\cite{HP:04} provide the result for the problems with $L_{max}$ and $\overline C$, while \cite{lendl} consider the case with $\overline{wC}$. In all four cases, Property \ref{prop::seq_old} is shown to hold and therefore, it is not necessary to consider idle time insertion to obtain optimal solutions.

\begin{prop}\label{prop:idle_prbs}
	An optimal solution to problems $1|\Delta_{max}\leq \epsilon|f_1$, $f_1\in \{\overline U, \overline{T}\}$, and $1|\overline \Delta\leq \epsilon|f_2$, $f_2\in \{L_{max}, \overline{wC}, \overline{U}, \overline{T}\}$ may contain inserted idle times.
\end{prop}
\begin{proof}
	To prove each of the stated results we provide, for each problem, a counterexample that shows the sub-optimality of a schedule with no idle time inserted.
	
	We start with the problem $1|\Delta_{max}\leq\epsilon|\overline U$.
	Consider the schedule $\pi^*$ in Figure \ref{fig:ex_deltamax_U}, with two old jobs $i$ and $j$, with $p_i=8,\, p_j=1$ and $d_i=8,\,d_j=9$ and two new jobs $k$ and $h$  with $p_k=p_h=1$ and due dates $d_k=d_h=8$. Let be $\epsilon=5$. An optimal solution for this instance is given by schedule $\sigma_{idle}$ with $\overline U=1$, and where the idle time enables to meet the disruption constraint. In contrast, restricting to the set of schedules without inserted idle times leads to the optimal solution $\sigma_{no-idle}$ with $\overline U=2$.
	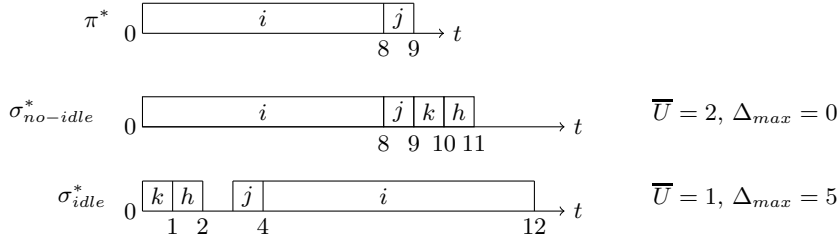
\begin{figure}[h]
		\centering
		\begin{tikzpicture}[x=0.4cm,y=0.4cm]
			\draw[xstep=1,ystep=1,ultra thin](-6,0)(25,1);
			\draw[](-1.5,0.5)node{\small $\pi^*$};
			\draw[](-0.4,0)node{\small $0$};
			\draw[](4,0.5)node{\small $i$};
			\draw[](8.5,0.5)node{\small $j$};
			\draw[](8,-0.5)node{\small $8$};
			\draw[](9,-0.5)node{\small $9$};
			\draw[->](0,0)--(10,0)node[right]{\small $t$};
			\draw[-](0,0)--(0,1);
			\draw[-](0,1)--(9,1);
			\draw[-](9,1)--(9,0);
			\draw[-](8,0)--(8,1);
		\end{tikzpicture} \\ [2ex]
		\begin{tikzpicture}[x=0.4cm,y=0.4cm]
			\draw[xstep=1,ystep=1,ultra thin](-6,0)(25,1);
			\draw[](-3,0.5)node{\small $\sigma^*_{no-idle}$};
			\draw[](-0.4,0)node{\small $0$};
			\draw[](9.5,0.5)node{\small $k$};
			\draw[](10.5,0.5)node{\small $h$};
			\draw[](8.5,0.5)node{\small $j$};
			\draw[](4,0.5)node{\small $i$};
			
			\draw[](20,0.5)node{\small $\overline U=2,\,\Delta_{max}=0$};
			\draw[->](0,0)--(14,0)node[right]{\small $t$};
			\draw[](0,1)rectangle(8,0)node[below]{\small 8};
			\draw[](8,1)rectangle(9,0)node[below]{\small 9};
			\draw[](9,1)rectangle(10,0)node[below]{\small 10};
			\draw[](10,1)rectangle(11,0)node[below]{\small 11};
		\end{tikzpicture}\\ [1ex]
		\begin{tikzpicture}[x=0.4cm,y=0.4cm]
			\draw[xstep=1,ystep=1,ultra thin](-6,0)(25,1);
			\draw[](-2,0.5)node{\small $\sigma^*_{idle}$};
			\draw[](-0.4,0)node{\small $0$};
			\draw[](0.5,0.5)node{\small $k$};
			\draw[](1.5,0.5)node{\small $h$};
			\draw[](3.5,0.5)node{\small $j$};
			\draw[](8,0.5)node{\small $i$};
			\draw[](20,0.5)node{\small $\overline U=1,\,\Delta_{max}=5$};
			\draw[](1,-0.5)node{\small $1$};
			\draw[](2,-0.5)node{\small $2$};
			\draw[](4,-0.5)node{\small $4$};
			\draw[](13,-0.5)node{\small $12$};
			\draw[->](0,0)--(14,0)node[right]{\small $t$};
			\draw[-](0,0)--(0,1);
			\draw[-](0,1)--(2,1);
			\draw[-](3,1)--(13,1);
			\draw[-](13,1)--(13,0);
			\draw[-](1,1)--(1,0);
			\draw[-](2,0)--(2,1);
			\draw[-](3,1)--(3,0);
			\draw[-](4,1)--(4,0);
		\end{tikzpicture}
		\caption{Idle time insertion for the $1|\Delta_{max}\leq\epsilon|\overline U$ problem}
		\label{fig:ex_deltamax_U}
	\end{figure}
	
	Next, we consider the problem $1|\Delta_{max}\leq\epsilon|\overline T$.
	Consider jobs $i,j,k,h$ from the above example and $\epsilon=5$. Consider an additional old job $\ell$ with $p_\ell=1$ and $d_\ell=10$. Schedule $\pi^*$ is shown in Figure \ref{fig:ex_deltamax_T}. The optimal schedule reaches $\overline T=5$. In contrast, restricting to the set of schedules without inserted idle times leads to the optimal solution $\sigma_{no-idle}$ with $\overline T=7$.
	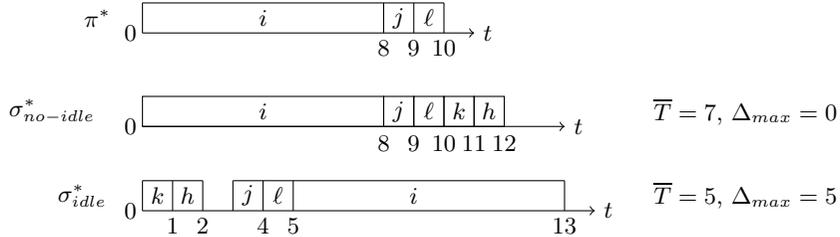
\begin{figure}[h]
		\centering
		\begin{tikzpicture}[x=0.4cm,y=0.4cm]
			\draw[xstep=1,ystep=1,ultra thin](-6,0)(25,1);
			\draw[](-1.5,0.5)node{\small $\pi^*$};
			\draw[](-0.4,0)node{\small $0$};
			\draw[](4,0.5)node{\small $i$};
			\draw[](8.5,0.5)node{\small $j$};
			\draw[](9.5,0.5)node{\small $\ell$};
			\draw[](8,-0.5)node{\small $8$};
			\draw[](9,-0.5)node{\small $9$};
			\draw[](10,-0.5)node{\small $10$};
			\draw[->](0,0)--(11,0)node[right]{\small $t$};
			\draw[-](0,0)--(0,1);
			\draw[-](0,1)--(10,1);
			\draw[-](10,1)--(10,0);
			\draw[-](8,0)--(8,1);
			\draw[-](9,0)--(9,1);
		\end{tikzpicture}\\ [2ex]
		\begin{tikzpicture}[x=0.4cm,y=0.4cm]
			\draw[xstep=1,ystep=1,ultra thin](-6,0)(25,1);
			\draw[](-3,0.5)node{\small $\sigma^*_{no-idle}$};
			\draw[](-0.4,0)node{\small $0$};
			\draw[](9.5,0.5)node{\small $\ell$};
			\draw[](10.5,0.5)node{\small $k$};
			\draw[](11.5,0.5)node{\small $h$};
			\draw[](8.5,0.5)node{\small $j$};
			\draw[](4,0.5)node{\small $i$};
			\draw[](20,0.5)node{\small $\overline T=7,\,\Delta_{max}=0$};
			\draw[->](0,0)--(14,0)node[right]{\small $t$};
			\draw[](0,1)rectangle(8,0)node[below]{\small 8};
			\draw[](8,1)rectangle(9,0)node[below]{\small 9};
			\draw[](9,1)rectangle(10,0)node[below]{\small 10};
			\draw[](10,1)rectangle(11,0)node[below]{\small 11};
			\draw[](11,1)rectangle(12,0)node[below]{\small 12};
		\end{tikzpicture}\\ [1ex]
		\begin{tikzpicture}[x=0.4cm,y=0.4cm]
			\draw[xstep=1,ystep=1,ultra thin](-6,0)(25,1);
			\draw[](-2,0.5)node{\small $\sigma^*_{idle}$};
			\draw[](-0.4,0)node{\small $0$};
			\draw[](0.5,0.5)node{\small $k$};
			\draw[](1.5,0.5)node{\small $h$};
			\draw[](3.5,0.5)node{\small $j$};
			\draw[](9,0.5)node{\small $i$};
			\draw[](4.5,0.5)node{\small $\ell$};
			\draw[](20,0.5)node{\small $\overline T=5,\,\Delta_{max}=5$};
			\draw[](1,-0.5)node{\small $1$};
			\draw[](2,-0.5)node{\small $2$};
			\draw[](4,-0.5)node{\small $4$};
			\draw[](5,-0.5)node{\small $5$};
			\draw[](14,-0.5)node{\small $13$};
			\draw[->](0,0)--(15,0)node[right]{\small $t$};
			\draw[-](0,0)--(0,1);
			\draw[-](0,1)--(2,1);
			\draw[-](3,1)--(14,1);
			\draw[-](14,1)--(14,0);
			\draw[-](1,1)--(1,0);
			\draw[-](2,0)--(2,1);
			\draw[-](3,1)--(3,0);
			\draw[-](4,1)--(4,0);
			\draw[-](5,1)--(5,0);
		\end{tikzpicture}
		\caption{Idle time insertion for the $1|\Delta_{max}\leq\epsilon|\overline T$ problem}
		\label{fig:ex_deltamax_T}
	\end{figure}
	
	Now, we turn to the problem $1|\overline \Delta\leq\epsilon|L_{max}$. Let be $i,j,h$ the jobs in $J^O$ with $p_i=8, p_j=p_h=1$ and $d_i=12,d_j=13,d_h=14$ and $k$ a new job with $p_k=7$ and $d_k=1$. Let us have $\epsilon=10$. As illustrated in Figure \ref{fig:ex_deltasum_lmax}, we assume that $\pi^*=(i,j,h)$. 
	the optimal schedule is $\sigma^*_{idle}=(k,j,h,i)$ with $L_{max}=10$ with an idle time of 1 unit before job $j$. In contrast, restricting to the set of schedules without inserted idle times leads to the optimal solution $\sigma_{no-idle}$ with $L_{max}=14$.
	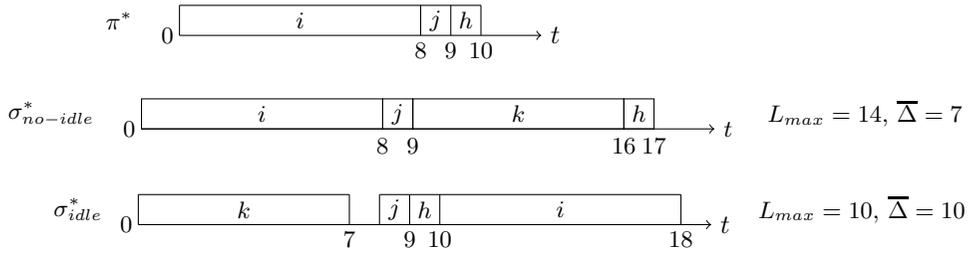
\begin{figure}[h]
		\centering
		\begin{tikzpicture}[x=0.4cm,y=0.4cm]
			\draw[xstep=1,ystep=1,ultra thin](-6,0)(25,1);
			\draw[](-2,0.5)node{\small $\pi^*$};
			\draw[](-0.4,0)node{\small $0$};
			\draw[](4,0.5)node{\small$i$};
			\draw[](8.5,0.5)node{\small $j$};
			\draw[](9.5,0.5)node{\small $h$};
			\draw[](8,-0.5)node{\small $8$};
			\draw[](9,-0.5)node{\small $9$};
			\draw[](10,-0.5)node{\small $10$};
			\draw[->](0,0)--(12,0)node[right]{$t$};
			\draw[-](0,0)--(0,1);
			\draw[-](0,1)--(10,1);
			\draw[-](8,0)--(8,1);
			\draw[-](9,0)--(9,1);
			\draw[-](10,0)--(10,1);	
		\end{tikzpicture} \\ [2ex]
		\begin{tikzpicture}[x=0.4cm,y=0.4cm]
			\draw[xstep=1,ystep=1,ultra thin](-6,0)(25,1);
			\draw[](-3,0.5)node{\small $\sigma^*_{no-idle}$};
			\draw[](-0.4,0)node{\small $0$};
			\draw[](12.5,0.5)node{\small $k$};
			\draw[](16.5,0.5)node{\small $h$};
			\draw[](8.5,0.5)node{\small $j$};
			\draw[](4,0.5)node{\small $i$};
			\draw[](24,0.5)node{\small $ L_{max}=14,\,\overline \Delta=7$};
			\draw[->](0,0)--(19,0)node[right]{\small $t$};
			\draw[](0,1)rectangle(8,0)node[below]{\small 8};
			\draw[](8,1)rectangle(9,0)node[below]{\small 9};
			\draw[](9,1)rectangle(16,0)node[below]{\small 16};
			\draw[](16,1)rectangle(17,0)node[below]{\small 17};
		\end{tikzpicture}\\ [2ex]
		\begin{tikzpicture}[x=0.4cm,y=0.4cm]
			\draw[xstep=1,ystep=1,ultra thin](-6,0)(25,1);
			\draw[](-2,0.5)node{\small $\sigma^*_{idle}$};
			\draw[](-0.4,0)node{\small $0$};
			\draw[](3.5,0.5)node{\small$k$};
			\draw[](8.5,0.5)node{\small $j$};
			\draw[](9.5,0.5)node{\small$h$};
			\draw[](14,0.5)node{\small $i$};
			\draw[](24,0.5)node{\small $L_{max}=10,\,\overline \Delta=10$};
			\draw[](7,-0.5)node{\small $7$};
			\draw[](9,-0.5)node{\small $9$};
			\draw[](10,-0.5)node{\small $10$};
			\draw[](18,-0.5)node{\small $18$};
			\draw[->](0,0)--(19,0)node[right]{$t$};
			\draw[-](0,0)--(0,1);
			\draw[-](0,1)--(7,1);
			\draw[-](8,1)--(18,1);
			\draw[-](18,1)--(18,0);
			
			\draw[-](7,0)--(7,1);
			\draw[-](8,0)--(8,1);
			\draw[-](9,0)--(9,1);
			\draw[-](10,0)--(10,1);
		\end{tikzpicture}
		\caption{Idle time insertion for the $1|\overline \Delta\leq\epsilon|L_{max}$ problem}
		\label{fig:ex_deltasum_lmax}
	\end{figure}

	Next, we consider the problem  $1|\overline \Delta\leq\epsilon|\overline{wC}$.
	Consider jobs $i,j,h\in J^O$ and $k\in J^N$ with $p_i=6,p_j=p_h=1,p_k=4$ and $w_i=7,w_j=w_h=1,w_h=12$. For $\epsilon=7$, the optimal schedule is $\sigma^*_{idle}=(k,j,h,i)$ with $\overline{wC}=139$ with an idle time of 1 unit before job $j$ (Figure \ref{fig:ex_deltasum_wc}). However, by imposing that no inserted idle time is allowed, the optimal solution becomes $\sigma^*_{no-idle}$ with $\overline{wC}=193$.
	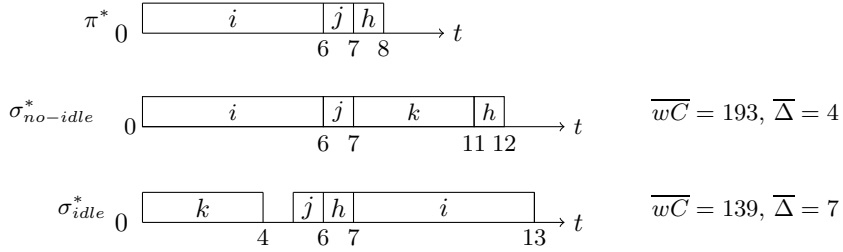
\begin{figure}[h]
		\centering
		\begin{tikzpicture}[x=0.4cm,y=0.4cm]
			\draw[xstep=1,ystep=1,ultra thin](-6,0)(25,1);
			\draw[](-1.5,0.5)node{\small $\pi^*$};
			\draw[](-0.7,0)node{$0$};
			\draw[](3,0.5)node{$i$};
			\draw[](6.5,0.5)node{$j$};
			\draw[](7.5,0.5)node{$h$};
			\draw[](6,-0.5)node{\small $6$};
			\draw[](7,-0.5)node{\small $7$};
			\draw[](8,-0.5)node{\small $8$};
			\draw[->](0,0)--(10,0)node[right]{$t$};
			\draw[-](0,0)--(0,1);
			\draw[-](0,1)--(8,1);
			\draw[-](6,0)--(6,1);
			\draw[-](7,0)--(7,1);
			\draw[-](8,0)--(8,1);	
		\end{tikzpicture} \\ [2ex]
		\begin{tikzpicture}[x=0.4cm,y=0.4cm]
			\draw[xstep=1,ystep=1,ultra thin](-6,0)(25,1);
			\draw[](-3,0.5)node{\small $\sigma^*_{no-idle}$};
			\draw[](-0.4,0)node{\small $0$};
			\draw[](9,0.5)node{\small $k$};
			\draw[](11.5,0.5)node{\small $h$};
			\draw[](6.5,0.5)node{\small $j$};
			\draw[](3,0.5)node{\small $i$};
			\draw[](20,0.5)node{\small $\overline {wC}=193,\,\overline \Delta=4$};
			\draw[->](0,0)--(14,0)node[right]{\small $t$};
			\draw[](0,1)rectangle(6,0)node[below]{\small 6};
			\draw[](6,1)rectangle(7,0)node[below]{\small 7};
			\draw[](7,1)rectangle(11,0)node[below]{\small 11};
			\draw[](11,1)rectangle(12,0)node[below]{\small 12};
		\end{tikzpicture}\\ [2ex]
		\begin{tikzpicture}[x=0.4cm,y=0.4cm]
			\draw[xstep=1,ystep=1,ultra thin](-6,0)(25,1);
			\draw[](-2,0.5)node{\small $\sigma^*_{idle}$};
			\draw[](-0.7,0)node{$0$};
			\draw[](2,0.5)node{$k$};
			\draw[](5.5,0.5)node{$j$};
			\draw[](6.5,0.5)node{$h$};
			\draw[](10,0.5)node{$i$};
			\draw[](20,0.5)node{\small $\overline {wC}=139,\,\overline \Delta=7$};
			\draw[](4,-0.5)node{\small $4$};
			\draw[](6,-0.5)node{\small $6$};
			\draw[](7,-0.5)node{\small $7$};
			\draw[](13,-0.5)node{\small $13$};
			\draw[->](0,0)--(14,0)node[right]{$t$};
			\draw[-](0,0)--(0,1);
			\draw[-](0,1)--(4,1);
			\draw[-](5,1)--(13,1);
			\draw[-](13,1)--(13,0);
			
			\draw[-](4,0)--(4,1);
			\draw[-](5,0)--(5,1);
			\draw[-](6,0)--(6,1);
			\draw[-](7,0)--(7,1);
			\draw[-](13,0)--(13,1);
		\end{tikzpicture}
		\caption{Idle time insertion for the $1|\overline \Delta\leq\epsilon|\overline{wC}$ problem}
		\label{fig:ex_deltasum_wc}
	\end{figure}

	Finally, we consider problems  $1|\overline \Delta\leq\epsilon|\{\overline{U},\overline T\}$.
	Set $J^O$ is made up of the three jobs $i,j,\ell$ with $p_i=8,\, p_j=p_\ell=1$ and $d_i=8,\,d_j=9,\,d_\ell=10$ and two new jobs $k$ and $h$  with $p_k=p_h=1$ and $d_k=d_h=8$.
	Let be $\epsilon=5$ and the optimal solution with $\overline U=1$ is given by schedule $\{k,h,j,\ell,i\}$ with one unit of idle time before $j$ (Figure \ref{fig:ex_deltasum_UT}). However, by imposing that no inserted idle time is allowed, the optimal solution becomes $\sigma_{no-idle}$ with $\overline U=2$. 
	
	The same counterexample holds for $\overline f=\overline T$, with a resulting objective $\overline T=5$ of an optimal schedule with idle time and $\overline T=7$ when imposing a constraint of no inserted idle time. 
	\begin{figure}[h]
		\centering
		\begin{tikzpicture}[x=0.4cm,y=0.4cm]
			\draw[xstep=1,ystep=1,ultra thin](-6,0)(25,1);
			\draw[](-1.5,0.5)node{\small $\pi^*$};
			\draw[](-0.4,0)node{\small $0$};
			\draw[](4,0.5)node{\small $i$};
			\draw[](8.5,0.5)node{\small $j$};
			\draw[](9.5,0.5)node{\small $l$};
			\draw[](8,-0.5)node{\small $8$};
			\draw[](9,-0.5)node{\small $9$};
			\draw[](10,-0.5)node{\small $10$};
			\draw[->](0,0)--(11,0)node[right]{\small $t$};
			\draw[-](0,0)--(0,1);
			\draw[-](0,1)--(10,1);
			\draw[-](10,1)--(10,0);
			\draw[-](8,0)--(8,1);
			\draw[-](9,0)--(9,1);
		\end{tikzpicture}\\ [2ex]
		\begin{tikzpicture}[x=0.4cm,y=0.4cm]
			\draw[xstep=1,ystep=1,ultra thin](-6,0)(25,1);
			\draw[](-3,0.5)node{\small $\sigma^*_{no-idle}$};
			\draw[](-0.4,0)node{\small $0$};
			\draw[](9.5,0.5)node{\small $\ell$};
			\draw[](10.5,0.5)node{\small $k$};
			\draw[](11.5,0.5)node{\small $h$};
			\draw[](8.5,0.5)node{\small $j$};
			\draw[](4,0.5)node{\small $i$};
			\draw[](20,0.5)node{\small $\overline {U}=2,\,\overline \Delta=0$};
			\draw[->](0,0)--(14,0)node[right]{\small $t$};
			\draw[](0,1)rectangle(8,0)node[below]{\small 8};
			\draw[](8,1)rectangle(9,0)node[below]{\small 9};
			\draw[](9,1)rectangle(10,0)node[below]{\small 10};
			\draw[](10,1)rectangle(11,0)node[below]{\small 11};
			\draw[](11,1)rectangle(12,0)node[below]{\small 12};
		\end{tikzpicture}\\ [2ex]
		\begin{tikzpicture}[x=0.4cm,y=0.4cm]
			\draw[xstep=1,ystep=1,ultra thin](-6,0)(25,1);
			\draw[](-2,0.5)node{\small $\sigma^*_{idle}$};
			\draw[](-0.4,0)node{\small $0$};
			\draw[](0.5,0.5)node{\small $k$};
			\draw[](1.5,0.5)node{\small $h$};
			\draw[](3.5,0.5)node{\small $j$};
			\draw[](9,0.5)node{\small $i$};
			\draw[](4.5,0.5)node{\small $l$};
			\draw[](20,0.5)node{\small $\overline {U}=1,\,\overline \Delta=5$};
			\draw[](1,-0.5)node{\small $1$};
			\draw[](2,-0.5)node{\small $2$};
			\draw[](4,-0.5)node{\small $4$};
			\draw[](5,-0.5)node{\small $5$};
			\draw[](14,-0.5)node{\small $13$};			
			\draw[->](0,0)--(15,0)node[right]{\small $t$};
			\draw[-](0,0)--(0,1);
			\draw[-](0,1)--(2,1);
			\draw[-](3,1)--(14,1);
			\draw[-](14,1)--(14,0);
			\draw[-](1,1)--(1,0);
			\draw[-](2,0)--(2,1);
			\draw[-](3,1)--(3,0);
			\draw[-](4,1)--(4,0);
			\draw[-](5,1)--(5,0);
		\end{tikzpicture}
		\caption{Idle time insertion for the $1|\overline \Delta\leq\epsilon|\overline{U}$ and $1|\overline \Delta\leq\epsilon|\overline{T}$ problems}
		\label{fig:ex_deltasum_UT}
	\end{figure}
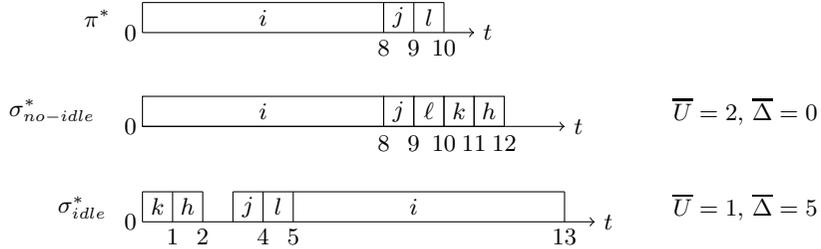
\end{proof}

Table \ref{tb:idleness} summarizes the above results: for each couple of objective function and disruption constraint, we indicate \textit{mit} (\textit{machine idle time}) when inserting machine idle times may be necessary to compute optimal solutions. We indicate \textit{n-mit} (\textit{no machine idle time}) for problems whose optimal solutions are without inserted idle times.
\begin{table}[ht]
	\centering
	\begin{tabular}{p{1cm}|p{1cm}p{1cm}}
		$f$ &$\Delta_{max}$ & $\overline \Delta$ \\[1ex]
		\hline 
		$L_{max}$   &n-mit  &mit\\[1ex]
		
		$\overline C$ &n-mit &n-mit\\[1ex]
		$\overline {wC}$ &n-mit  &mit\\[1ex]
		$\overline U$ &mit  &mit\\[1ex]
		$\overline T$   &mit  &mit\\
		\hline		
	\end{tabular}
	\caption{Problem classification w.r.t. the insertion of machine idle times.}
	\label{tb:idleness}
\end{table}

\section{Complexity results}

In this section, we recall complexity results of rescheduling problems from the literature and provide proofs for some open problems. Table \ref{tb:complexity} summarizes the current state of known complexity results. The first two rows contain the results stated by \cite{HP:04}. The complexity of problem $1|\Delta_{max}\leq \epsilon|\overline{wC}$ is due to \cite{lendl}.
The complexity status in bold are these established in this section.\\
\begin{table}[ht]
	\centering
	\begin{tabular}{l|cc}
		$f$&$\Delta_{max}$ & $\overline \Delta$ \\[2ex]
		\hline 
		$L_{max}$   & $O(n + n_N log(n_N) )$  &\begin{tabular}{c}
			strongly\\$\cal NP$-hard
		\end{tabular}\\[2ex]
		$\overline C$ & $O(n + n_N log(n_N) )$ &\begin{tabular}{c}
			$\cal NP$-hard
		\end{tabular}\\[2ex]
		$\overline {wC}$ &\begin{tabular}{c}
			$\cal NP$-hard
		\end{tabular}  &\begin{tabular}{c}\textbf{$\cal NP$-hard}
		\end{tabular} \\[2ex]
		$\overline U$ & \begin{tabular}{c}\textbf{$\cal NP$-hard}
		\end{tabular}&  \begin{tabular}{c}			\textbf{strongly}\\\textbf{$\cal NP$-hard}
		\end{tabular}\\[2ex]
		$\overline T$   & \begin{tabular}{c}	\textbf{$\cal NP$-hard}
		\end{tabular} & \begin{tabular}{c}			\textbf{strongly}\\\textbf{$\cal NP$-hard}
		\end{tabular}\\
		\hline		
	\end{tabular}
	\caption{\centering Computational complexity of rescheduling problems.}
	\label{tb:complexity}
\end{table}

\begin{thm}
	Problem $1|\Delta_{max}\leq \epsilon|\overline U$ is $\cal NP$-hard.
\end{thm}
\begin{proof}
	Consider the 2-PARTITION problem, where a set $S$ of elements $a_1, \dots, a_n$ is given and the question is to determine whether there exists or not a partition into two subsets $S_1, S_2$, such that $\sum_{i\in S_1}a_i=\sum_{i\in S_2}a_i=B$. The problem is known to be NP-complete in the weak sense.
	
	To show that the rescheduling problem is $\cal NP$-hard, we show that 2-PARTITION can be reduced to the decision version of the rescheduling problem denoted by $1|\Delta_{max}\leq \epsilon,\,\overline{U}\leq 1|-$.
	Consider the following instance of the decision version of the rescheduling problem:\\
	$J^O =\{i_1, i_2 \}$, with
	\begin{itemize}
		\item $p_{i1}=2B+5$
		\item $p_{i2}=1$
		\item $d_{i1}=2B+5$
		\item $d_{i2}=2B+6$
	\end{itemize}
	$J^N = \{j_1, j_2\} \cup J$, where $J$ is a set of $n$ jobs, with
	\begin{itemize}
		\item $p_{j1}=p_{j2}=1$
		\item $d_{j1}=d_{j2}=2$
		\item $p_j=a_j,\, \forall j\in J$
		\item $d_j = 4B+8,\, \forall j\in J$
	\end{itemize}
	Finally, we set $\epsilon=B+3$.
	
	Schedule $\pi^*$ is constructed as follows to have both old jobs on time.
	\begin{figure}[h]
		\centering
		\begin{tikzpicture}[x=0.4cm,y=0.4cm]
			\draw[xstep=1,ystep=1,ultra thin](-3,0)(22,1);
			\draw[](-2,0.5)node{ $\pi^*$};
			\draw[](-0.7,0)node{ $0$};
			\draw[](12,-0.7)node{ $2B+6$};
			\draw[-](12,0.2)--(12,-0.2);
			\draw[](5,0.5)node{ $i_1$};
			\draw[](11.5,0.5)node{ $i_2$};
			\draw[->](0,0)--(15,0)node[right]{};
			\draw[-](0,0)--(0,1);
			\draw[-](0,1)--(12,1);
			\draw[-](11,0)--(11,1);
			\draw[-](12,0)--(12,1);
		\end{tikzpicture}
	\end{figure}
	\newline If there exists a solution to 2-PARTITION, there exist two index sets $J_1$ and $J_2$, such that $\sum_{i\in J_1}a_i=\sum_{i\in J_2}a_i=B$.
	\begin{figure}[h]
		\centering
		\begin{tikzpicture}[x=0.4cm,y=0.4cm]
			\draw[xstep=1,ystep=1,ultra thin](-3,0)(22,1);
			\draw[](-2,0.5)node{$\sigma$};
			\draw[](-0.7,0)node{ $0$};
			\draw[](3.5,-0.7)node{ $B$};
			\draw[](18.5,-0.7)node{ $B$};
			\draw[](0.5,0.5)node{ $j_1$};
			\draw[](1.5,0.5)node{ $j_2$};
			\draw[](3.5,0.5)node{ $j\in J_1$};
			\draw[](5.5,0.5)node{ $i_2$};
			\draw[](11.5,0.5)node{ $i_1$};
			\draw[](18.5,0.5)node{ $j \in J_2$};
			
			\draw[->](0,0)--(21,0)node[right]{};
			\draw[-](0,0)--(0,1);
			\draw[-](0,1)--(20,1);
			\draw[-](1,0)--(1,1);
			\draw[-](2,0)--(2,1);
			\draw[-](5,0)--(5,1);
			\draw[-](6,0)--(6,1);
			\draw[-](17,0)--(17,1);	
			\draw[-](20,0)--(20,1);
		\end{tikzpicture}
	\end{figure}
	Then, the following schedule $\sigma$ solves the rescheduling problem $1|\Delta_{max}\leq \epsilon|\overline{U}\leq 1$,
	where job $i_1$ is the only tardy job. 
	Also, $$\Delta_{i1}=(2B+5)+2+1+B-(2B+5)=B+3\leq \epsilon$$ and $$\Delta_{i2}=(2B+5)+1-2-1-B=B+3\leq \epsilon$$.\\
	
	On the other side, if the rescheduling problem has a solution with $\overline{U}\leq 1$, the two jobs $j_1$ and $j_2$ must be scheduled at the beginning positions of the schedule, otherwise at least two jobs will become tardy. But then, $i_1$ and $i_2$ must be swapped and, to meet the constraint on the disruption, $i_2$ must complete not earlier than $B+3$ and $i_1$ not later than $3B+8$. Given that, since the due date of jobs $j,\, \forall j \in J$, equals the sum of all jobs, the solution cannot contain idle times and the only way to get that is to have a partition of the jobs in two sets $J_1, J_2$ such that $\sum_{i\in J_1}p_i = \sum_{i\in J_2}p_i = B$. 
	
\end{proof} 

\begin{thm}
	Problems $1|\overline \Delta\leq \epsilon|\overline U$ and $1|\overline \Delta\leq \epsilon|\overline T$ are strongly $\cal NP$-hard.
\end{thm}
\begin{proof}
	Consider the 3-PARTITION problem, where a set $S$ of elements $a_1,\dots,a_{3t}$ is given. A solution to the problem exists if the elements can be partitioned into sets $S_1,\dots,S_{t}$, such that for every set $S_j, j=1,\dots,t$, there is an integer $y$, such that $\sum_{i\in S_j}a_i=y$ and $|S_j|=3$. The problem is known to be NP-complete in the strong sense.
	
	We show that the rescheduling problem $1|\overline \Delta\leq \epsilon|\overline{U}$ is strongly $\cal NP$-hard, we follow the line of the proof of  \cite{HP:04}. Given any instance of 3-PARTITION, we show that the problem reduces to the decision version of $1|\overline \Delta\leq \epsilon|\overline{U}$ denoted by $1|\overline \Delta\leq \epsilon,\,\overline{U}\leq 0|-$.
	
	Consider the following instance of the rescheduling problems with $2t$ old jobs and $3t$ new jobs.
	\begin{itemize}
		\item $J^O=\{1,\dots,2t\}$
		\item $J^N =\{2t+1,\dots,5t\}$
		\item $p_j=1, j=1,\dots,t$
		\item $p_j=ty, j=t+1,\dots,2t$
		\item $d_j=t(2ty+1), j=t+1,\dots,2t$
		\item $p_j=ta_{j-2t}, j=2t+1,\dots,5t$
		\item $d_j=t^2y+t-1, j=2t+1,\dots,5t$
		\item $\epsilon=t^3y+t(t+1)/2$
	\end{itemize}
	
	It is shown below, that a solution to $1|\overline \Delta\leq \epsilon,\,\overline{U}\leq 0|-$ exists if and only if a solution of 3-PARTITION exists.
	
	Schedule $\pi^*$ is constructed as follows.
	\begin{figure}[h]
		\centering
		\begin{tikzpicture}[x=0.3cm,y=0.4cm]
			\draw[xstep=1,ystep=1,ultra thin](-3,0)(20,1);
			\draw[](-2,0.5)node{ $\pi^*$};
			\draw[](-0.7,0)node{$0$};
			\draw[](7,-0.7)node{$ty+1$};
			\draw[-](7,0.2)--(7,-0.2);
			\draw[](17,-0.7)node{$t(ty+1)$};
			\draw[-](17,0.2)--(17,-0.2);
			\draw[](3,0.5)node{$t+1$};
			\draw[](6.5,0.5)node{$1$};
			\draw[](8.5,0.5)node{$\dots$};
			\draw[](13,0.5)node{$2t$};
			\draw[](16.5,0.5)node{$t$};
			\draw[->](0,0)--(20,0)node[right]{};
			\draw[-](0,0)--(0,1);
			\draw[-](0,1)--(7,1);
			\draw[-](10,1)--(17,1);
			\draw[-](6,0)--(6,1);
			\draw[-](7,0)--(7,1);
			\draw[-](10,0)--(10,1);	
			\draw[-](16,0)--(16,1);
			\draw[-](17,0)--(17,1);
		\end{tikzpicture}
	\end{figure}
	\newline If a solution to 3-PARTITION exists, there exist index sets $J_1,J_2,\dots,J_t$, such that for any $J_j$, $\sum_{i\in J_j}a_i=y$ and $|J_j|=3$. Then,
	the following schedule solves $1|\overline \Delta\leq \epsilon,\,\overline U\leq 0|-$.
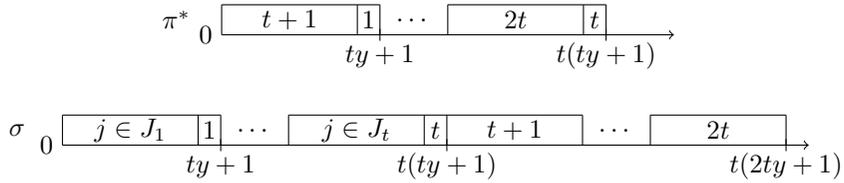
	\begin{figure}[h]
		\centering
		\begin{tikzpicture}[x=0.3cm,y=0.4cm]
			\draw[xstep=1,ystep=1,ultra thin](-3,0)(33,1);
			\draw[](-2,0.5)node{$\sigma$};
			\draw[](-0.7,0)node{$0$};
			\draw[](7,-0.7)node{$ty+1$};
			\draw[-](7,0.2)--(7,-0.2);
			\draw[](17,-0.7)node{$t(ty+1)$};
			\draw[-](17,0.2)--(17,-0.2);
			\draw[](32,-0.7)node{$t(2ty+1)$};
			\draw[-](32,0.2)--(32,-0.2);
			
			\draw[](3,0.5)node{$j\in J_1$};
			\draw[](6.5,0.5)node{$1$};
			\draw[](8.5,0.5)node{$\dots$};
			\draw[](13,0.5)node{$j\in J_t$};
			\draw[](16.5,0.5)node{$t$};
			\draw[](20,0.5)node{$t+1$};
			\draw[](24.5,0.5)node{$\dots$};
			\draw[](29,0.5)node{$2t$};
			\draw[->](0,0)--(33,0)node[right]{};
			\draw[-](0,0)--(0,1);
			\draw[-](0,1)--(7,1);
			\draw[-](10,1)--(23,1);
			\draw[-](26,1)--(32,1);
			\draw[-](6,0)--(6,1);
			\draw[-](7,0)--(7,1);
			\draw[-](10,0)--(10,1);	
			\draw[-](16,0)--(16,1);
			\draw[-](17,0)--(17,1);
			\draw[-](23,0)--(23,1);
			\draw[-](26,0)--(26,1);
			\draw[-](32,0)--(32,1);		
		\end{tikzpicture}
	\end{figure}
	\newline On the other hand, in order to have a schedule with $\overline U\leq 0$, all jobs $t+1,\dots,2t$ have to be scheduled after the new jobs, otherwise one of them would be late. 
	Suppose that exactly $h$ jobs of the set $\{1,\dots,t\}$ are scheduled after the first job of $\{t+1,\dots,2t\}$ in $\sigma^*$. Each such job completes in $\sigma^*$ no earlier than time $t^2y+ty+(t-h)+1$, whereas it completes in $\pi^*$ no later than time $t(ty+1)$. Moreover, jobs $t+1, \dots,2t$ are completed at times $ty,(2ty+1),\dots,(t^2y+t-1)$ in $\pi^*$ and are preceded by all new jobs and $t-h$ jobs of $\{1,\dots,t\}$ and therefore are completed no earlier than times $(t^2y+(t-h)+ty),(t^2y+(t-h)+2ty),\dots,(t^2y+(t-h)+t^2y)$ in $\sigma^*$. Computing the total disruption, we obtain:\\
	$\overline \Delta \geq t^3y+t(t-h)-t(t-1)/2+h(ty-h+1)$\\
	$=\epsilon+h(ty-h+1)$.\\
	Since $y\geq 3$, by definition of 3-PARTITION, we conclude that $h=0$. Also, in order to have the disruption constraint satisfied, jobs $1,\dots,t$ must be scheduled on time with respect to their initial completion times and the new jobs in sets $S_1,\dots,S_t$ between them as shown in the figure above. The partition of the new jobs solves 3-PARTITION.
	
	The same reduction also works for problem $1|\overline \Delta\leq \epsilon|\overline{T}$ and shows that 3-PARTITION reduces to the decision problem $1|\overline \Delta\leq \epsilon,\,\overline{T}\leq 0|-$.
	
	In order to conclude the proof, we show that the problem is not a number problem. The first condition requires the existence of a polynomial $\lambda$ of both the size of the rescheduling problem and the size of 3-PARTITION that bounds the largest number of the input, as follows:
	$$ \max(\max_{j\in J^O \cup J^N}(p_j,d_j),\epsilon) \leq 2t^3y+t^2+t \leq \lambda(2(|J^O|+|J^N|)+1,t) \leq \lambda(10t+1,t)$$
	Since 3-PARTITION is not a number problem, there exists a polynomial $\lambda'$ such that $y \leq \lambda'(3t)$ which implies the existence of $\lambda$. \\
	The second condition requires the existence of a polynomial $\rho$  of the length of the size of the rescheduling instance that bounds the size of 3-PARTITION, as follows:
	$$ t \leq \rho(10t+1)$$
	Clearly, the statement holds and shows that the rescheduling problem is not a number problem.
\end{proof}

\begin{thm}
	Problem $1|\overline \Delta\leq \epsilon|\overline{wC}$ is $\cal NP$-hard.
\end{thm}
\begin{proof}
	Consider the special case with one old job $j$. Clearly, with one old job we have $\overline \Delta =\Delta_{j}=\Delta_{max}$. Since problem $1|\Delta_{max}\leq \epsilon|\overline{wC}$ has been proven to be $\cal NP$-hard in the weak sense by \cite{lendl}, even when there is a single old job, problem $1|\overline \Delta\leq \epsilon|\overline{wC}$ is $\cal NP$-hard. 
\end{proof}

\begin{thm}
	Problem $1|\Delta_{max}\leq \epsilon|\overline T$ is $\cal NP$-hard.
\end{thm}
\begin{proof}
	Consider the well-known generalization of the problem, i.e. the scheduling problem $1||\overline T$, which is known to be weakly $\cal NP$-hard. It is enough to take $\epsilon=+\infty$ to show that $1||\overline T$ is a particular case of $1|\Delta_{max}\leq \epsilon|\overline T$, i.e. the rescheduling problem is $\cal NP$-hard. 
\end{proof}

\section{Conclusions}
We considered the set of rescheduling problems for new orders on a single machine, that arise from the combination of several scheduling objective functions and two disruption criteria. The new jobs have to be integrated into the original optimal schedule of old jobs to minimize the scheduling objective function, subject to a constraint on the disruption. We considered the most commonly-used scheduling objective functions, i.e. total completion time, total weighted completion time, total number of late jobs, total tardiness and maximum lateness, and two disruption criteria, i.e. the total and maximum absolute time deviation of old jobs.

The work provides a complete classification of problems according to their complexity and need for machine idle time insertion.
In particular, the different behaviour of the two disruption criteria and the relevance of whether or not to preserve the original order of the old jobs with respect to the insertion of idle time are highlighted. 

Future research should focus on the development of algorithms for problems that may require the insertion of idle times. 

\bibliographystyle{apalike}
\bibliography{Resched}

\end{document}